\documentclass{article}
\usepackage{graphicx}
\usepackage{amsmath, amsthm, latexsym}
\usepackage{amssymb}
\newtheorem{theorem}{Theorem}
\newtheorem{lemma}{Lemma}

\newtheorem{corollary}{Corollary}

\newtheorem{remark}{Remark}
\usepackage{color}

\begin{document}

\vspace*{3cm} \thispagestyle{empty}
\vspace{5mm}

\noindent \textbf{\Large Exact Fermi coordinates for a class of spacetimes}\\

\noindent  \textbf{\normalsize David Klein}\footnote{Department of Mathematics, California State University, Northridge, Northridge, CA 91330-8313. Email: david.klein@csun.edu.}
\textbf{\normalsize and Peter Collas}\footnote{Department of Physics and Astronomy, California State University, Northridge, Northridge, CA 91330-8268. Email: peter.collas@csun.edu.}\\

\vspace{4mm} \parbox{11cm}{\noindent{\small We find exact Fermi coordinates for timelike geodesic observers for a class of spacetimes that includes anti-de Sitter spacetime, de Sitter spacetime, the constant density interior Schwarzschild spacetime with positive, zero, and negative cosmological constant, and the Einstein static universe.  Maximal charts for Fermi coordinates are discussed.}\vspace{5mm}\\
\noindent {\small KEY WORDS: Fermi coordinates, de Sitter space, anti-de Sitter space, interior Schwarzschild metric, Einstein static universe, cosmological constant, Gaussian curvature, sectional curvature, Jacobi field}\\

\noindent MSC numbers: 83C15, 83C10\\
\noindent PACS numbers: 04.20.Cv, 04.20.Jb, 04.20.-q}\\
\vspace{6cm}
\pagebreak

\setlength{\textwidth}{27pc}
\setlength{\textheight}{43pc}
\noindent \textbf{{\normalsize 1. Introduction}}\\

\noindent The effects of a gravitational field are most naturally analyzed by using a system of locally inertial coordinates.  For an observer following a timelike path, Fermi-Walker coordinates provide such a system. A Fermi-Walker coordinate frame is nonrotating in the sense of Newtonian mechanics and is realized physically as a system of gyroscopes \cite{walker, synge, MTW, CK2}.   Applications of these coordinate systems are extensive and include the study of tidal dynamics, gravitational waves, relativistic statistical mechanics, and quantum gravity  \cite{CM, Ishii, Marz, FG, CK, KC2, B, PP82}.  In the case that the path of the observer is geodesic, Fermi-Walker coordinates are commonly referred to as Fermi or Fermi normal coordinates.  The metric in that case is Minkowskian to first order near the path, with second order corrections involving only the curvature tensor \cite{MM63}.\\  

\noindent Under general conditions, a timelike path has a neighborhood on which a Fermi-Walker coordinate system can be defined \cite{oniell} (p. 200).  In addition, general formulas in the form of Taylor expansions for coordinate transformations to and from Fermi-Walker coordinates, valid in some neighborhood of a given timelike path in general spacetimes, were given in \cite{KC1}. However, to the best of our knowledge, rigorous results for the radius of a tubular neighborhood of a timelike path, for the domain of Fermi coordinates are not available. In addition to potential applications, it is therefore revealing to find examples where exact coordinate transformations to and from Fermi coordinates can be calculated in order to determine the maximum extent of coordinate charts for those coordinate systems.\\

 \noindent In this paper, we find exact transformations to and from Fermi coordinates for a class of spacetimes. Our starting point is a generic metric given by Eq.\eqref{general} below. In Sect. 2, Theorems \ref{theorem1} and  \ref{theorem2} give explicit charts with Fermi coordinates for metrics of the form of Eq.\eqref{general}.  We use sectional curvature of an appropriate $2$-dimensional submanifold to define Jacobi fields that measure the separation of (Fermi) coordinate, spacelike geodesics.  Our examples, described in Sect. 3, include the metrics for anti-de Sitter spacetime (AdS$_{4}$), de Sitter spacetime (dS$_{4}$), the interior constant density Schwarzschild spacetime with postive, negative, or zero cosmological constant, and the Einstein static universe.  We also discuss the breakdown of Fermi coordinates at the horizon in dS$_{4}$.  Concluding Remarks are given in Sect. 4.\\

\noindent \textbf{{\normalsize 2. Fermi Coordinates and Curvature for a Class of Metrics}}\\

\noindent In a spacetime $M$, let $\sigma(\tau)$ be a a timelike geodesic parameterized by proper time $\tau$ with unit tangent vector $e_{0}(\tau)$.  A Fermi normal coordinate system along $\sigma$ is determined by an orthonormal tetrad of vectors, $e_{0}(\tau)$, $e_{1}(\tau), e_{2}(\tau), e_{3}(\tau)$ parallel along $\sigma$. Fermi coordinates $x^{0}$, $x^{1}$, $x^{2}$, $x^{3}$ relative to this tetrad   are defined by,

\begin{equation}\label{F2}
\begin{split}
x^{0}\left (\exp_{\sigma(\tau)} (\lambda^{j}e_{j}(\tau)\right)&= \tau \\
x^{k}\left (\exp_{\sigma(\tau)} (\lambda^{j}e_{j}(\tau)\right)&= \lambda^{k}, 
\end{split} 
\end{equation}

\noindent where here and below, Greek indices run over $0,1,2,3$ and Latin over $1,2,3$. The exponential map, $\exp_{p}(\vec{v})$, denotes the evaluation at affine parameter $1$ of the geodesic starting at the point $p$ in the spacetime, with initial derivative $\vec{v}$, and it is assumed that the $\lambda^{j}$ are sufficiently small so that the exponential maps in Eq.\eqref{F2} are defined.\\

\noindent Consider a line element of the form,

\begin{align}
\begin{split}\label{general}
ds^{2}=&-\left(1-f(x,y,z)\right)dt^{2}+dx^{2}+dy^{2}+dz^{2}\\ 
&+\left[\left(1- k r^{2}\right)^{-1}-1\right]dr^{2},
\end{split}
\end{align}

\noindent where  $r^{2}=x^{2}+y^{2}+z^{2}$, $k$ is a constant, and $f(x,y,z)$ is a smooth function, which together with its first partial derivatives, vanishes at $x=y=z=0$.  When $f(x,y,z)\equiv 0=k$, Eq.\eqref{general} is the Minkowski metric.  Although not essential, we assume for convenience that $1-f(x,y,z)$ does not vanish when $1- k r^{2} >0$, and that this last expression determines the range of spatial coordinates $(x,y,z)$ for the chart on which the metric is described by Eq.\eqref{general}.\\

\noindent Since all first partial derivatives of the metric elements determined by Eq.\eqref{general} vanish on the timelike path $\sigma(t)=(t,0,0,0)$, it immediately follows that the connection coefficients also vanish on $\sigma(t)$, and that $\sigma(t)$ is a geodesic.  Moreover, $t=\tau$ is proper time, and the following orthonormal tetrad is parallel along along $\sigma(t)$:

\begin{equation}
\begin{split}\label{tetrad}
\frac{\partial}{\partial t}&=e_{0}(\tau)=(1,0,0,0)\\
\frac{\partial}{\partial x}&=e_{1}(\tau)=(0,1,0,0)\\
\frac{\partial}{\partial y}&=e_{2}(\tau)=(0,0,1,0)\\
\frac{\partial}{\partial z}&=e_{3}(\tau)=(0,0,0,1)\\
\end{split}
\end{equation}

\noindent We construct Fermi coordinates for $\sigma(t)=(t,0,0,0)$,  begining with the inverse transformation, from Fermi coordinates $\{x^{0}, x^{1}, x^{2},x^{3}\}$ to Cartesian coordinates $\{t, x,y,z\}$, given by the following theorem.\\

\noindent In what follows, it is convenient to define $a \equiv\sqrt{ |k|}>0$.\\

\begin{theorem}
\label{theorem1}\textup{\textbf{(a)}} When $k>0$, the transformation from Fermi coordinates along $\sigma(t)$ to the coordinates $\{t, x,y,z\}$ is given by,
\begin{align}
t&=x^{0},\label{0+}\\\nonumber\\
x&=x^{1}\left(\frac{\sin(\rho a)}{\rho a}\right)\,,\label{1} \\\nonumber\\
y&=x^{2}\left(\frac{\sin(\rho a)}{\rho a}\right)\,,\label{2} \\\nonumber\\
z&=x^{3}\left(\frac{\sin(\rho a)}{\rho a}\right)\,,\label{3}
\end{align}

\noindent\textup{\textbf{(b)}} When $k<0$, the transformation from Fermi coordinates along $\sigma(t)$ to the coordinates $\{t, x,y,z\}$ is given by,
\begin{align}
t&=x^{0},\label{3+}\\\nonumber\\
x&=x^{1}\left(\frac{\sinh(\rho a)}{\rho a}\right)\,,\label{4} \\\nonumber\\
y&=x^{2}\left(\frac{\sinh(\rho a)}{\rho a}\right)\,,\label{5} \\\nonumber\\
z&=x^{3}\left(\frac{\sinh(\rho a)}{\rho a}\right)\,,\label{6}
\end{align}

\noindent where $\rho^{2}= (x^{1})^{2}+(x^{2})^{2}+(x^{3})^{2}$.
\end{theorem}

\begin{proof} It follows from Eq.\eqref{F2} that a necessary and sufficient condition for  $\{x^{0},$ $ x^{1},x^{2},x^{3}\}$ to be Fermi coordinates relative to a tetrad $e_{0}(\tau)$, $e_{1}(\tau)$, $e_{2}(\tau)$, $ e_{3}(\tau)$ along a geodesic $\sigma$ is that in these coordinates,

\begin{equation}
\exp_{\sigma(\tau)} (sa^{j}e_{j}(\tau))= (\tau, sa^{1}, sa^{2}, sa^{3}), 
\end{equation}

\noindent where $s$ measures proper distance and $\sqrt{(a^{1})^{2}+(a^{2})^{2}+(a^{3})^{2}}=1$. Thus, it suffices in our case to prove that $X_{t}(s) \equiv (t, sa^{1}, sa^{2}, sa^{3})$ is geodesic in the coordinates $\{x^{0},$ $ x^{1},x^{2},x^{3}\}$ given by Eqs.\eqref{0+}--\eqref{3} for $k>0$ and \eqref{3+}--\eqref{6} for $k<0$. This is readily verified by using these equations to transform the metric of Eq.\eqref{general}, yielding the results of Corollary \ref{corollary1} below, from which the connection coefficients are determined.  It then follows by direct calculation that,

\begin{equation}\label{geodesic2}
\Gamma^{\nu}_{ij}(t,x^{1},x^{2},x^{3})x^{i}x^{j}=0,
\end{equation}

\noindent which is equivalent to,
 \begin{equation}\label{geodesic}
\frac{\,d^{2}X^{\nu}}{ds^{2}} + \Gamma^{\nu}_{\alpha \beta}\frac{\,dX^{\alpha}}{ds}\frac{\,dX^{\beta}}{ds}
=\Gamma^{\nu}_{ij}(t,sa^{1},sa^{2},sa^{3})a^{i}a^{j}=0.
\end{equation}

\noindent Thus, $X_{t}(s) \equiv (t, sa^{1}, sa^{2}, sa^{3})$ is geodesic for all choices of $(a^{1}, a^{2}, a^{3})$.
\end{proof}

\noindent The following two corollaries follow from Theorem \ref{theorem1} and direct calculation.\\

\begin{corollary}
\label{corollary1} The metric in Fermi coordinates for the observer $\sigma(t)$, \textup{\textbf{(a)}} when $k>0$ is given by,

\begin{align}
g_{00}&=-\left[1-f\left(x^{1}\left[\frac{\sin(\rho a)}{\rho a}\right], x^{2}\left[\frac{\sin(\rho a)}{\rho a}\right],x^{3}\left[\frac{\sin(\rho a)}{\rho a}\right]\right)\right]\label{g00+}, \\\nonumber\\
g_{0i}&=0\,,\label{g0i+} \\\nonumber\\
g_{ij}&=\frac{x^{i}x^{j}}{\rho^{2}}+\frac{\sin^{2}\left( a\rho\right)}{a^{2}\rho^{2}}\left(\delta_{ij}-\frac{x^{i}x^{j}}{\rho^{2}}\right).\label{gij+}
\end{align}

\noindent\textup{\textbf{(b)}} when $k<0$, is given by,
\begin{align}
g_{00}&=-\left[1-f\left(x^{1}\left[\frac{\sinh(\rho a)}{\rho a}\right], x^{2}\left[\frac{\sinh(\rho a)}{\rho a}\right],x^{3}\left[\frac{\sinh(\rho a)}{\rho a}\right]\right)\right]\label{g00-},\\\nonumber\\
g_{0i}&=0\,,\label{g0i-} \\\nonumber\\
g_{ij}&=\frac{x^{i}x^{j}}{\rho^{2}}+\frac{\sinh^{2}\left( a\rho\right)}{a^{2}\rho^{2}}\left(\delta_{ij}-\frac{x^{i}x^{j}}{\rho^{2}}\right).\label{gij-}
\end{align}
\end{corollary}

\begin{corollary}
\label{corollary2} Under the change of spatial coordinates, $x^{1}=\rho \sin\theta \cos\phi$, $x^{2}=\rho \sin\theta \sin\phi$, $x^{3}=\rho \cos\theta$, the Fermi metric given by Corollary \ref{corollary1} \textup{\textbf{(a)}} for $k>0$ becomes,

\begin{equation}\label{polar+}
ds^{2}=g_{00}dt^{2}+d\rho^{2}+\frac{\sin^{2}(a\rho)}{a^{2}}(d\theta^{2} + \sin^{2} \theta \,d\phi^{2}),
\end{equation}

\noindent \textup{\textbf{(b)}} for $k<0$ becomes,

\begin{equation}\label{polar-}
ds^{2}=g_{00}dt^{2}+d\rho^{2}+\frac{\sinh^{2}(a\rho)}{a^{2}}(d\theta^{2} + \sin^{2} \theta \,d\phi^{2}),
\end{equation}

\noindent where $g_{00}$ is given by Eq.\eqref{g00+} for part \textup{(a)}, and \eqref{g00-} for part \textup{(b)}.\\
\end{corollary}

\begin{theorem}
\label{theorem2}\textup{\textbf{(a)}} When $k>0$, the transformation from the coordinates $\{t, x,y,z\}$ to Fermi coordinates along $\sigma(t)$ is given by,

\begin{align}
x^{0}&=t,\\\nonumber\\
x^{1}&=x\left(\frac{\sin^{-1}(r a)}{r a}\right)\,,\label{1'} \\\nonumber\\
x^{2}&=y\left(\frac{\sin^{-1}(r a)}{r a}\right)\,,\label{2'} \\\nonumber\\
x^{3}&=z\left(\frac{\sin^{-1}(r a)}{r a}\right)\,,\label{3'}
\end{align}

\noindent\textup{\textbf{(b)}}  When $k<0$, the transformation from the coordinates $\{t, x,y,z\}$ to Fermi coordinates along $\sigma(t)$ is given by,

\begin{align}
x^{0}&=t,\\\nonumber\\
x^{1}&=x\left(\frac{\sinh^{-1}(r a)}{r a}\right)\,,\label{4'} \\\nonumber\\
x^{2}&=y\left(\frac{\sinh^{-1}(r a)}{r a}\right)\,,\label{5'} \\\nonumber\\
x^{3}&=z\left(\frac{\sinh^{-1}(r a)}{r a}\right)\,.\label{6'}
\end{align}
\end{theorem}
\noindent where, as above, $r^{2}=x^{2}+y^{2}+z^{2}$.

\begin{proof} To prove part (a), observe that squaring and adding Eqs.\eqref{1},\eqref{2}, and \eqref{3}, gives,

\begin{equation}\label{r2}
r^{2}= \frac{\sin^{2}(\rho a)}{a^{2}}
\end{equation}

\noindent Solving for $\rho$ in the above equation, and then for $x,y$, and $z$ in Eqs.\eqref{1},\eqref{2}, and \eqref{3} yields Eqs \eqref{1'},\eqref{2'}, and \eqref{3'}.  The proof of part (b) using,

\begin{equation}\label{r2-}
r^{2}= \frac{\sinh^{2}(\rho a)}{a^{2}}
\end{equation}

\noindent is similar.
\end{proof}

\begin{remark}
\label{indep}
The independence from the function $f(x,y,z)$ of the coordinate transformations appearing in Theorems \ref{theorem1} and \ref{theorem2} is a consequence of Eq.\eqref{geodesic2} and the assumption that $f(x,y,z)$ and its first partial derivatives vanish on $\sigma$.
\end{remark}

\begin{remark}
\label{range}
Under the assumptions made in the paragraph below Eq.\eqref{general}, it follows from Eqs. \eqref{r2} and \eqref{r2-} that for $k>0$, the domain of the spatial Fermi coordinates may be chosen to include any open set in which $0\leqslant\rho<\pi/2a$, and for $k<0$, $0\leqslant\rho<\infty$.
\end{remark}

\noindent The following corollary will be used to identify a Jacobi field for the congruence of spatial geodesics orthogonal to the Fermi observer's world line.\\

\begin{corollary}\label{curvature}
Let $M$ be a spacetime with metric given by Eq.\eqref{polar+} or \eqref{polar-}. Let $N$ be a $2$-dimensional submanifold of $M$ generated by the Fermi coordinates $t$ and $\rho$ with the angular coordinates held fixed so that the induced metric on $N$ is given by,

\begin{equation}\label{submanifold}
ds^{2}=g_{00}dt^{2}+d\rho^{2}.
\end{equation}

\noindent Then the Gaussian curvature  $K$ of $N$ is given by,

\begin{equation}\label{gauss}
K=\frac{-1}{\sqrt{-g_{00}}}\,\frac{\partial^{2}}{\partial\rho^{2}}\sqrt{-g_{00}}.
\end{equation}
\end{corollary}

\begin{proof}  The result follows easily from Proposition 44 (p. 81) of \cite{oniell} and direct calculation.
\end{proof}

\begin{remark}\label{sectional}
In the case that $g_{00}$ is a function of $\rho$ only, it is easy to verify that $N$ (with the induced metric, Eq.\eqref{submanifold}) is totally geodesic in $M$, i.e., the shape tensor vanishes.  Thus, the intrinsic geometry of $N$ coincides with its extrinsic geometry as a submanifold of $M$. In particular, the sectional curvature of $N$ in $M$ is the Gaussian curvature $K$. 
\end{remark}

\noindent We assume now that $g_{00}$ is a function of $\rho$ only, i.e.,

\begin{equation}\label{rho}
g_{00}=g_{00}(\rho)
\end{equation}

\noindent The vector field $\partial/\partial t$ is a variation vector field for the geodesic variation of spacelike geodesics of the form, $X_{t}(\rho)=(t, \rho)$, parameterized in $N$ by $t$. Therefore the Jacobi equation,

\begin{equation}\label{jacobi1}
\nabla_\frac{\partial}{\partial \rho}\nabla_{\frac{\partial}{\partial \rho}}(\partial/\partial t)=R_{\frac{\partial}{\partial t}\frac{\partial}{\partial \rho}}(\partial/\partial \rho),
\end{equation}
 
\noindent is satisfied, where $\nabla$ is the Levi-Civita connection (on either $N$ or $M$) and $R$ is the Riemann curvature operator.  In light of Remark \ref{sectional}, the right side of Eq.\eqref{jacobi1} may be expressed in terms of the Gaussian curvature K, yielding,

\begin{equation}\label{jacobi2}
\nabla_\frac{\partial}{\partial \rho}\nabla_{\frac{\partial}{\partial \rho}}(\partial/\partial t)=-K\, \partial/\partial t.
\end{equation}

\noindent The unit vector $T= \frac{1}{\sqrt{-g_{00}}} \frac{\partial}{\partial t}$ 
is orthogonal to $\partial/\partial \rho$ and thus parallel along the spacelike geodesic $X_{t}(\rho)=(t, \rho)$ (with $t$ fixed).  It follows that,

\begin{equation}
\nabla_\frac{\partial}{\partial \rho}\nabla_{\frac{\partial}{\partial \rho}}(\partial/\partial t)= \nabla_\frac{\partial}{\partial \rho}\nabla_{\frac{\partial}{\partial \rho}}(\sqrt{-g_{00}}\,T)= (\frac{\partial^{2}}{\partial\rho^{2}}\sqrt{-g_{00}})T.
\end{equation}

\noindent Eq.\eqref{jacobi2} then becomes,

\begin{equation}
\left(\frac{\partial^{2}}{\partial\rho^{2}}\sqrt{-g_{00}}+K \sqrt{-g_{00}}\right)T=0,
\end{equation}

\noindent which is equivalent to Eq.\eqref{gauss}.  Thus, $y=(t_{2}-t_{1})\sqrt{-g_{00}}$ is a measure of separation of the spacelike geodesics $X_{t_{1}}(\rho)=(t_{1}, \rho)$ and $X_{t_{2}}(\rho)=(t_{2}, \rho)$ at proper distance $\rho$, and is a solution of the initial value problem,

\begin{equation}
\begin{split}\label{DE1}
\frac{\partial^{2}y}{\partial\rho^{2}}&+K(\rho) y=0\\
y'(0)&=(t_{2}-t_{1})\frac{-g'_{00}(0)}{\sqrt{-g_{00}(0)}}=0\\
y(0)&=(t_{2}-t_{1})\sqrt{-g_{00}(0)}=t_{2}-t_{1},
\end{split}
\end{equation}

\noindent where, in the initial data, we have used the assumptions on $g_{00}$ that immediately follow Eq.\eqref{general}, and for convenience, we assume that $t_{2} > t_{1}$.\\

\noindent The following lemma shows that when the Gaussian curvature on $N$ is nonpositive, there is a natural timelike separation of the spacelike geodesics that define the Fermi space coordinate, which never becomes null.

\begin{lemma}\label{DE2}Let $K(\rho) \leqslant 0$ be continuous  and suppose that $y$ is a solution to the initial value problem, Eqs.\eqref{DE1}.  Then $y$ has no positive roots.
\end{lemma}  

\begin{proof} Suppose to the contrary that $\rho_{0}$ is the least positive root of $y$.  Then $y'(\rho_{0})\leqslant0$.  Since by assumption, $y'(0)=0$,
$y'(\rho)$ must be a decreasing function on some open subinterval of 
$[0,\rho_{0}]$. On that subinterval, $y''(\rho)<0$, which contradicts the assumption on $K$.
\end{proof}

\noindent \textbf{{\normalsize 3. Examples}}\\

\noindent Using the results of the previous section, we find in this section exact expressions for the metrics in Fermi coordinates along particular timelike geodesics in AdS$_{4}$,dS$_{4}$, the interior constant density Schwarzschild spacetime with positive, zero, and negative cosmological constant, and the Einstein static universe.  We also discuss the range of Fermi coordinates together with the Gaussian curvatures of the associated submanifolds (i.e., $N$) described in the previous section.\\

\noindent \textbf{Example 1.}  AdS$_{4}$ and dS$_{4}$ metrics in Fermi coordinates\\

\noindent In static coordinates of dS$_{4}$, or the analog for AdS$_{4}$, the metric is,

\begin{equation}
ds^{2}=-\left(1-\frac{\Lambda r^{2}}{3}\right)dt^{2}+r^{2}(d\theta^{2} + \sin^{2} \theta \,d\phi^{2})+\left(1-\frac{\Lambda r^{2}}{3}\right)^{-1}dr^{2},
\label{0}
\end{equation}

\noindent where the cosmological constant $\Lambda$ is positive in the case of dS$_{4}$, and negative for AdS$_{4}$.  In the case of dS$_{4}$, Eq.\eqref{0}
is singular at the cosmological horizon where $r = \sqrt{3/ \Lambda}$.  The horizon divides spacetime into four regions as may be seen from the Penrose diagram \cite{hawking}. In one of these regions  the timelike Killing vector $\partial/\partial t$ is future-directed, $0\leqslant r < \sqrt{3/ \Lambda}$, and an observer at $r = 0$ is surrounded by the cosmological horizon  at $r = \sqrt{3/ \Lambda}$.  For the case of dS$_{4}$, we consider the Fermi observer at $r=0$ in this causal region.\\ 

\noindent By contrast, when $\Lambda <0$ (for AdS$_{4}$), the range of $r$ is unrestricted, i.e., $0\leqslant r < \infty$.  In either case, Eq.\eqref{0} may be rewritten as,

\begin{align}
\begin{split}\label{0'}
ds^{2}=&-dt^{2}+r^{2}(d\theta^{2} + \sin^{2} \theta \,d\phi^{2})+dr^{2}\\ 
&+\frac{\Lambda r^{2}}{3}c^{2}dt^{2}+\left[\left(1-\frac{\Lambda r^{2}}{3}\right)^{-1}-1\right]dr^{2}
\end{split}
\end{align}

\noindent The first line of Eq.\eqref{0'} is the Minkowski metric in spherical coordinates.  Changing to Cartesian space coordinates $x,y,z$, and identifying $r^{2}=x^{2}+y^{2}+z^{2}$, Eq.\eqref{0'} becomes,

\begin{align}
\begin{split}\label{0''}
ds^{2}=&-\left(1-\frac{\Lambda r^{2}}{3}\right)dt^{2}+dx^{2}+dy^{2}+dz^{2}\\ 
&+\left[\left(1-\frac{\Lambda r^{2}}{3}\right)^{-1}-1\right]dr^{2},
\end{split}
\end{align}

\noindent which has the form of Eq.\eqref{general} with $f(x,y,z)=\Lambda r^{2}/3$, $k=\Lambda/3$.\\ 

\noindent Using Eq.\eqref{r2-}, we find that the Fermi metric for the observer $\sigma(t)=(t,0,0,0)$ in AdS$_{4}$ is,

\begin{equation}\label{ads}
ds^{2}= -\cosh^{2}\left( a\rho\right)dt^{2} + g_{ij}dx^{i}dx^{j},
\end{equation}

\noindent where  $a = \sqrt{|k|}=\sqrt{ |\Lambda|/3}$ and the spatial metric coefficients $g_{ij}$ are given by Eq.\eqref{gij-}. Fermi coordinates $\{x^{0}, x^{1}, x^{2},x^{3}\}$ are global on the covering space for AdS$_{4}$, and consistent with Remark \ref{range}, Eq.\eqref{ads} is valid on the entire spacetime. The associated polar metric given by Corollary \ref{corollary2}, though heretofore not associated with Fermi coordinates, is independently well-known and extant in the literature:

\begin{equation}\label{polar-2}
ds^{2}= -\cosh^{2}\left( a\rho \right) dt^{2}+d\rho^{2}+\frac{\sinh^{2}(a\rho)}{a^{2}}(d\theta^{2} + \sin^{2} \theta \,d\phi^{2}).
\end{equation}

\noindent The Fermi metric for the observer $\sigma(t)=(t,0,0,0)$ in static coordinates in dS$_{4}$ is analogous. Using Eq.\eqref{r2} for $\Lambda>0$,

\begin{equation}\label{ds}
ds^{2}= -\cos^{2}\left( a\rho\right)dt^{2} + g_{ij}dx^{i}dx^{j},
\end{equation}

\noindent where $a = \sqrt{k}=\sqrt{ \Lambda/3}$ and the spatial metric coefficients, $g_{ij}$ are given by Eq.\eqref{gij+}. Consistent with Remark \ref{range}, Fermi coordinates $\{x^{0}, x^{1}, x^{2},x^{3}\}$ cover the region of dS$_{4}$ satisfying $\rho= \sqrt{(x^{1})^{2}+(x^{2})^{2}+(x^{3})^{2}}<\pi/2a$, the same region covered by static coordinates, up to the cosmological horizon.  The associated polar metric given by Corollary \ref{corollary2} is,

\begin{equation}\label{polar+2}
ds^{2}= -\cos^{2}\left( a\rho \right) dt^{2}+d\rho^{2}+\frac{\sin^{2}(a\rho)}{a^{2}}(d\theta^{2} + \sin^{2} \theta \,d\phi^{2}).
\end{equation}

\begin{remark}
\label{CM}
We note that Eqs.\eqref{ds} and \eqref{polar+2} for dS$_{4}$ are not new. Chicone and Mashhoon, starting with a different coordinate system for the de Sitter universe, previously derived Eqs.\eqref{ds} and \eqref{polar+2} in \cite{CM}, and observed that Eq.\eqref{polar+2} appears in de Sitter's original 1917 investigations.  Exact Fermi coordinates for G\"odel spacetime are also given in \cite{CM}.
\end{remark}

\noindent With the notation of Corollary \ref{curvature}, a short calculation shows that the Gaussian curvature $K$ of the submanifold spanned by the Fermi coordinates $t,\rho$ with the angular coordinates held fixed is given by,

\begin{equation}\label{gauss1}
K=\frac{\Lambda}{3},
\end{equation}

\noindent so that $K$ is positive on the submanifold $N$ of dS$_{4}$ and negative on the corresponding submanifold of AdS$_{4}$.\\

\noindent Eqs.\eqref{DE1} apply to these examples, but it is instructive to analyze directly the way in which the Fermi coordinate system breaks down at the horizon of dS$_{4}$, where $\rho =\pi/2a$.  Consider two spacelike geodesics with the same fixed angular coordinates, orthogonal to the Fermi observer's worldline. Without loss of generality we take the angular coordinates to be fixed at $\phi=0$, and $\theta=\pi/2$ and the Fermi time coordinates to be $t_{1}$ and $t_{2}$ with $t_{1}<t_{2}$.  The two spacelike geodesics are then given by,

\begin{equation}\label{horizon1}
X_{i}(\rho)= (t_{i}, \rho, \pi/2, 0)\quad i=1,2.
\end{equation}

\noindent  When $\rho=0$, $X_{1}$ and $X_{2}$ lie on the timelike geodesic path of the Fermi observer.  For $0<\rho_{0}<\pi/2a$, the two spacetime points $X_{1}(\rho_{0})$ and $X_{2}(\rho_{0})$ are the same proper distance $\rho_{0}$ from the Fermi observer's path and are connected to each other by the timelike path,

\begin{equation}\label{horizon2}
\gamma_{\rho_{0}}(t)=(t,\rho_{0},\pi/2,0)\quad \text{for} \quad t_{1}\leqslant t \leqslant t_{2}.
\end{equation}

\noindent The observer following the path $\gamma_{\rho_{0}}(t)$ starts at $X_{1}(\rho_{0})$, waits for the fixed Fermi coordinate time interval, $t_{2}-t_{1}$ (without changing Fermi space coordinates), and then arrives at the spacetime point $X_{2}(\rho_{0})$.\\

\noindent However, the proper time along $\gamma_{\rho_{0}}(t)$ is less than the Fermi time interval by a factor of $\cos(a\rho_{0})$, which decreases to zero as $\rho_{0}\rightarrow \pi/2a$.  Expressed another way, the tangent vector $\partial/\partial t$ of $\gamma_{\rho_{0}}(t)$ becomes null at the horizon, $\rho= \pi/2a$.  Since the metric is Lorentzian, this alone is not enough to conclude that the two spacelike geodesics intersect at $\rho= \pi/2a$.  This is because of the possibility that that $\gamma_{\rho_{0}}(t)$ becomes a lightlike path, but does not degenerate to a single spacetime point.  However, the point $p$ of intersection can be identified via a different coordinate system, such as Kruskal coordinates, used for other purposes in \cite{hawking}. Thus, the Fermi coordinate patch cannot include points in the horizon or beyond.\\

\noindent \textbf{Example 2.}  Fermi coordinates for the Einstein static universe\\

\noindent The metric for the Einstein static universe may be written  (c.f. \cite{tolman}) as,

\begin{equation}\label{Einstein}
ds^{2}=-dt^{2}+r^{2}(d\theta^{2} + \sin^{2} \theta \,d\phi^{2})+\left(1-\frac{r^{2}}{R^{2}}\right)^{-1}dr^{2},
\end{equation}

\noindent where $R$ is a constant that depends on energy density and the cosmological constant.  Topologically, the spacetime is $\mathbb{R} \times S^{3}_{R}$, where $R$ is the radius of the $3$-sphere $S^{3}_{R}$. The same calculation leading to Eq.\eqref{0''} shows that this metric may be rewritten as, 

\begin{equation}\label{Einstein2}
ds^{2}=-dt^{2}+dx^{2}+dy^{2}+dz^{2}\\ 
+\left[\left(1-\frac {r^{2}}{R^{2}}\right)^{-1}-1\right]dr^{2},
\end{equation}

\noindent which has the form of Eq.\eqref{general} with $f(x,y,z) \equiv 0$ and $k=R^{-2}$ (and hence $a = R^{-1}$).  Thus, the Fermi metric for the observer $\sigma(t)=(t,0,0,0)$ is, 

\begin{equation}\label{EinsteinFermi}
ds^{2}= -dt^{2} + g_{ij}dx^{i}dx^{j},
\end{equation}

\noindent where the spatial metric coefficients $g_{ij}$ are given by Eq.\eqref{gij+}. The associated polar metric given by Corollary \ref{corollary2} is,

\begin{equation}\label{Einsteinpolar}
ds^{2}=-dt^{2}+d\rho^{2}+R^{2}\sin^{2}\left(\frac{\rho}{R}\right)(d\theta^{2} + \sin^{2} \theta \,d\phi^{2}),
\end{equation}

\noindent a known form of the metric \cite{exact}.  It follows trivally from Eq.\eqref{gauss} that the curvature $K=0$.  Consistent with Remark \ref{range}, if the range of $r$ in Eq.\eqref{Einstein} is $0\leqslant r<R$, then the corresponding range of the proper distance $\rho$ is given by $0\leqslant\rho< \pi R/2$ in Eqs.\eqref{EinsteinFermi} and \eqref{Einsteinpolar}.  However, as expected for the case that $K\leqslant0$, Fermi coordinates may be extended beyond this range to cover the entire spacetime, with the exception of the pole opposite to the origin or coordinates.  Thus, we may take the range of $\rho$ to be given by $0\leqslant\rho< \pi R$.\\

\noindent \textbf{Example 3.} Fermi coordinates for the interior constant density Schwarzschild spacetime with cosmological constant\\

\noindent The metric for a constant density fluid may be written as,

\begin{equation}\label{IS1}
ds^{2}=-A(r)d\bar{t}^{2}
+B(r)dr^{2}+r^{2}\left(d\theta^{2}+\sin^{2}\theta d\phi^{2}\right),
\end{equation}

\noindent where $M$ is the mass of the spherical fluid, $\Lambda$ is the cosmological constant, $R$ is the radial coordinate for the radius of the fluid and,

\begin{align}
\begin{split}\label{IS2}
A(r)=&\left[\frac{(3-R_{0}^{2}\Lambda)}{2}\sqrt{1-\frac{R^{2}}{R_{0}^{2}}}-\frac{(1-R_{0}^{2}\Lambda)}{2}\sqrt{1-\frac{r^{2}}{R_{0}^{2}}}\right]^{2},\\
B(r)=&\left(1-\frac{r^{2}}{R_{0}^{2}}\right)^{-1}.
\end{split}
\end{align}

\noindent Here,

\begin{equation}
R_{0}^{2}=\frac{3R^{3}}{6M+\Lambda R^{3}}.
\end{equation}

\noindent We assume that $A(r), B(r)$, and $R_{0}$ are all positive for $0\leqslant r \leqslant R$ so that the metric is well-defined. In order to find the metric form of Eq.\eqref{IS1} in Fermi coordinates, we first make a change of variable, $t=\sqrt{A(0)}\,\bar{t}$, with the space coordinates held fixed. Eq.\eqref{IS1} then becomes,

\begin{equation}\label{IS3}
ds^{2}=-\left(1-f(x,y,z)\right)dt^{2}
+B(r)dr^{2}+r^{2}\left(d\theta^{2}+\sin^{2}\theta d\phi^{2}\right),
\end{equation}

\noindent where,

\begin{equation}\label{fschwarz}
f(x,y,z)=1-\frac{A(r)}{A(0)}.
\end{equation}

\noindent The same calculation leading to Eq.\eqref{0''} shows that this metric may be rewritten as, 

\begin{equation}\label{IS4}
ds^{2}=-\left(1-f(x,y,z)\right)dt^{2}+dx^{2}+dy^{2}+dz^{2}\\ 
+\left[B(r)-1\right]dr^{2},
\end{equation}

\noindent which has the form of Eq.\eqref{general} with $k=R_{0}^{-2}>0$. Thus, the Fermi metric for the observer $\sigma(t)=(t,0,0,0)$ is, 

\begin{equation}\label{IS5}
ds^{2}= -\frac{A(r(\rho))}{A(0)}dt^{2} + g_{ij}dx^{i}dx^{j},
\end{equation}

\noindent where the spatial metric coefficients $g_{ij}$ are given by Eq.\eqref{gij+} with $a=1/R_{0}$, and where $r(\rho)^{2}$ is given by Eq.\eqref{r2}. The interval of values for $\rho$ corresponding to $0\leqslant r \leqslant R$ is $0\leqslant\rho\leqslant R_{0}\sin^{-1} (R/R_{0})$.  The associated polar metric given by Corollary \ref{corollary2} is,

\begin{equation}\label{ISpolar}
ds^{2}=-\frac{A(r(\rho))}{A(0)}dt^{2}+d\rho^{2}+\frac{\sin^{2}(a\rho)}{a^{2}}(d\theta^{2} + \sin^{2} \theta \,d\phi^{2}).
\end{equation}\\

\noindent The Gaussian curvature of the submanifold $N$ generated by the Fermi coordinates $t,\rho$, given by Eq.\eqref{gauss}  is,

\begin{equation}\label{gauss2}
K=-\frac{1-R_{0}^{2}\Lambda}{2R_{0}^{2}\sqrt{A(r(\rho))}}\cos(a\rho).
\end{equation}

\noindent It is clear that $K\leqslant0$, and by Lemma \ref{DE2}, orthogonal spacelike geodesics with different Fermi time coordinates remain temporally separated for $\rho\leqslant\pi/2a$.  The restriction of $\rho$ to smaller values, noted above, is a requirement of Buchdahl type inequalities \cite{boehmer, hiscock}. \\

\noindent \textbf{{\normalsize 4. Concluding Remarks}}\\

\noindent Using the results of Sect. 1, we have found Fermi coordinates in Cartesian and polar forms, for natural observers in AdS$_{4}$, dS$_{4}$, the Einstein static universe, and the interior Schwarzschild solution with cosmological constant. A Jacobi field measuring the separation of coordinate spacelike geodesics was described in terms of Gaussian curvature (or sectional curvature) of  $2$-dimensional submanifolds defined in terms of Fermi time and distance.\\

\noindent A breakdown of Fermi coordinates occurs when two or more spacelike 
geodesics, orthogonal to the Fermi observer's worldline $\sigma(\tau)$, and originating from that worldline at two different proper times, intersect at some spacetime point.  This occurs for dS$_{4}$ at the horizon for the Fermi observer.  In the other examples considered here, the charts for Fermi coordinates are global. In the case of the Einstein static universe, Fermi coordinates extend beyond the range of the coordinates used to define the metric given by Eq.\eqref{Einstein}.   We note that it is not difficult to construct additional examples of spacetimes with exact transformation formulas to Fermi coordinates (using Theorems \ref{theorem1} and \ref{theorem2}) by combining these examples so as to obtain Fermi coordinates for Schwarzschild-(anti) de Sitter space with interior constant density fluid.  The Fermi observer in those cases remains for all proper times at the center of the fluid.\\


\begin{thebibliography}{99}

\bibitem{walker} Walker, A. G.: Note on relativistic mechanics  \textit{Proc. Edin. Math. Soc.} \textbf{4}, 170-174 (1935).

\bibitem{synge}  Synge, J. L.: \textit{Relativity: The General Theory} North Holland, Amsterdam (1960).

\bibitem{MTW} Misner, C. W., Thorne, K. S., and Wheeler, J. A.  \textit{Gravitation}, W. H. Freeman, San Francisco, (1973) p. 329.

\bibitem{CK2} Collas, P., Klein, D.: A Simple Criterion for Nonrotating Reference Frames, \textit{Gen. Rel. Grav.} \textbf{36}, 1493-1499  (2004)

\bibitem{CM} Chicone, C., Mashhoon, B.: Explicit Fermi coordinates and tidal dynamics in de Sitter and G\"odel spacetimes \textit{Phys. Rev.} D \textbf{74},  064019 (2006).

\bibitem{Ishii} Ishii, M., Shibata, M., Mino, Y.: Black hole tidal problem in the Fermi normal coordinates  \textit{Phys. Rev.} D \textbf{71},  044017 (2005).

\bibitem{Marz} Marzlin, K-P.: Fermi coordinates for weak gravitational fields  \textit{Phys. Rev.} D \textbf{50}, 888-891 (1994).

\bibitem{FG} Fortini, P. L., Gualdi, C.: Fermi normal co-ordinate system and electromagnetic detectors of gravitational waves. I - Calculation of the metric  \textit{Nuovo Cimento} B \textbf{71}, 37-54 (1982).

\bibitem{CK} Collas, P., Klein, D., A Statistical mechanical problem in Schwarzschild spacetime \textit{Gen. Rel. Grav.} \textbf{39}, 737-755 (2007).

\bibitem{KC2} Klein, D., Collas, P.: Timelike Killing fields and relativistic statistical mechanics, {\it Class. Quantum Grav.}  \textbf{26}, 045018 (16 pp) arXiv:0810.1776v2 [gr-qc] (2009).

\bibitem{B} Bimonte, G., Calloni, E., Esposito, G., Rosa, L.:  Energy-momentum tensor for a Casimir apparatus in a weak gravitational field \textit{Phys. Rev.} D \textbf{74}, 085011 (2006).

\bibitem{PP82} Parker, L., Pimentel, L. O.:  Gravitational perturbation of the hydrogen spectrum \textit{Phys. Rev.} D \textbf{25}, 3180-3190 (1982).

\bibitem{MM63} Manasse, F. K., Misner, C. W.: Fermi normal coordinates and some basic concepts in differential geometry \textit{J. Math. Phys.} \textbf{4}, 735-745 (1963).

\bibitem{oniell} O'Neill, B.: \textit{Semi-Riemannian geometry with applications to relativity} (1983).  Academic Press, New York.


\bibitem{KC1}  Klein, D., Collas, P.: General Transformation Formulas for Fermi-Walker Coordinates \textit{Class. Quant. Grav.}  \textbf{25}, 145019 (17pp) DOI:10.1088/0264-9381/25/14/145019, [gr-qc] arxiv.org/abs/0712.3838v4 (2008).


\bibitem{hawking} Gibbons, G. W., Hawking, S. W.: Cosmological event horizons, thermodynamics, and particle creation, \textit{Phys. Rev.} D \textbf{15} 2738-2751 (1977).
\bibitem{tolman} Tolman, R. C.: Relativity, thermodynamics, and cosmology. Clarendon, Oxford (1934) p. 335.
\bibitem{exact} Stephani, H., Kramer, D., MacCallum, M., Hoenselaers, C., and Herlt, D.  \textit{Exact solutions to Einstein's field equations}, Second Edition, Cambridge University Press, (2003) p. 177.

\bibitem{boehmer} B\"{o}hmer, C.G., Harko, T.: Does the cosmological constant imply the existence of a minimum mass? \textit{Phys. Lett. B} \textbf{630}, 73-77 (2005) [arXiv:gr-qc/0509110] 

\bibitem{hiscock} Hiscock, W. A.: General relativistic fluid spheres with nonzero vacuum energy density \textit{J. Math. Phys.} \textbf{29}, 443-445 (1988)



\end{thebibliography}
\end{document}